\documentclass[sigconf]{acmart}

\AtBeginDocument{%
  \providecommand\BibTeX{{%
    \normalfont B\kern-0.5em{\scshape i\kern-0.25em b}\kern-0.8em\TeX}}}

\usepackage{booktabs} 
\usepackage{algorithm}
\usepackage{algorithmic}
\usepackage{caption}
\usepackage{subcaption}

\acmConference[Arxiv'22] {Advertising Conference}{2022}{Silicon Valley, CA, USA}

\begin{document}

\title{Do not Waste Money on Advertising Spend:  Bid Recommendation via Concavity Changes}

\author{{Deguang Kong}{*}, Konstantin Shmakov and Jian Yang} 
\authornote{This work was done during the authors were working at Yahoo Research.  The
correspondence should be addressed to doogkong@gmail.com.  The views and conclusions contained in
this document are those of the author(s) and should not be
interpreted as representing the official policies, either
expressed or implied, of any companies. }
\affiliation{%
  \institution{Yahoo Research,   San Jose, California, U.S.A, 94089}
  }

  \email{ doogkong@gmail.com, kshmakov@yahooinc.com,  jianyang@yahooinc.com}

\begin{abstract}

In computational advertising, a challenging problem is how to recommend the bid for advertisers to achieve the best return on investment (ROI) given budget constraint. This paper presents a bid recommendation scenario that discovers the concavity changes in click prediction curves. The recommended bid is derived based on the turning point from  significant increase (i.e. concave downward) to slow increase (convex upward). Parametric learning based method is applied by solving the corresponding constraint optimization problem. 
%
%
Empirical studies on real-world advertising scenarios clearly demonstrate the performance gains for business metrics (including revenue increase, click increase and advertiser ROI increase). 
\end{abstract}

\maketitle

\section{Introduction}
In computational advertising business, 
the publishers sell the ad inventories to the advertisers who win the auctions in ad exchange based on generalized second-price (GSP) ~\cite{10.1257/aer.97.1.242}
in the non-guaranteed delivery of ad\footnote{ {``Ad'' and ``advertisement'' are used interchangeably with the same meaning.} } in sales channel. 
The actual delivery of ad impressions depends on the bidding auctions of advertisers on demand side  as well as the total number of market supplies provided by the publishers. 
In online advertising, one major complaint of advertisers is: 
\emph{
``Half the money I spend on advertising is wasted; the trouble is I don't know which half\footnote{The original words are from John Wanamaker, who is  a proponent of advertising and a pioneer in marketing.}. "  
}
As an advertiser, one always needs to consider the spending question, such as ``am I spending enough on advertising?"  and ``am I spending too much on advertising" ?  On one hand, the advertiser will miss the opportunities if one bids at lower price. On the other hand, the advertiser will throw advertising dollars if one spends too much. Therefore, it is important to forecast how the advertiser performance looks like at different bid prices, making it easier for advertisers to make informed decision and minimize the risks of ad spending. 
\begin{figure}[t]
	\centering
	\includegraphics[height=1.45in,width=3.3in]{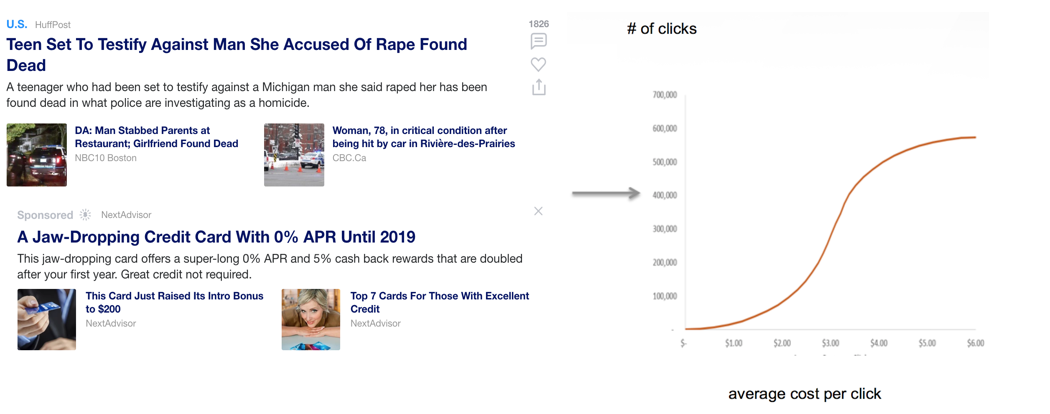}
	\caption{\small Market behavior of advertisers: number of clicks vs. average cost per click.}
	\label{fig:ads_plot}
\end{figure}

One key observation in ad market is:  advertisers' return on investment has non-linearity characteristics due to diminishing returns after the advertiser's bid reaches to a certain point which in fact leads to saturation effect.  For example, if the advertiser bids at \$3.60, it may get around 450k clicks; however, if the advertiser increases the bid to \$6.00, the true clicks are only 600k which is far less than two thirds more than the clicks at \$3.60 (shown in Fig.~\ref{fig:ads_plot}). The diminishing return is known as 
``saturation effect'' due to the reach of the limit of click where even higher bidder can have very minor incremental effect.  

Motivated from the above observations, in this paper, we propose a logistic growth inflection model  for bid optimization.  In this new strategy, the bid optimization aims to achieve the maximum return increase using the predicted bid landscape models, which aims to discover the significant change (a.k.a a turning point) in the key performance indicator (KPI) function at which the sign of the curvature changes.  Logistic growth function is used to capture the non-linear non-decreasing relations between KPI (e.g., number of clicks) and bid price. %
In particular, at low levels of ad bid, there is very little click response increase. As bid continues to increase, there is more significant positive growth of clicks until an inflection point. After that, an incremental dollar of bid increase will not produce the same increase of clicks any more due to the reaching of the stage of diminishing returns.  %
The designed new bid optimization strategy can be applied to plan the campaigns for advertisers to gain valuable insights, as well as providing informed decisions for real-time bidding  suggestions. 
%
To summarize, this paper makes the following contributions. 

$\bullet$ To maximize the return of click increases over each dollar spent increase on the ads, we propose a logistic growth inflection model to model the relations between number of clicks and bid price for bid recommendation. To the best of our knowledge, this is the \emph{first} work that proposes logistic growth model to solve the bid optimization problem. 

$\bullet$ To effectively solve the resultant optimization problem, we use parametric learning to fit the logistic growth function giving he sparse and noisy observations by minimizing the least-square errors, and then we solve the optimization problem with guaranteed global optimal solution. 


$\bullet$ Extensive experiment demonstrates the effectiveness of achieving the maximum return of click increase over the same bid increase, and it has clear advantages over other bidding strategies. The proposed bid recommendation scenario brings in 15.37\% bid increase,  30.24\% click increase and 14.50\% ROI increase over the baselines in empirical online test. 


\section{Problem Statement}
\label{sec:ps}



{\bf Background of Advertising modeling}. In computational advertising, an advertiser bids at certain prices to win impressions in auctions and therefore getting clicks after display.  The advertisers need to pay effective-cost-per-mille
(eCPM cost) or cost-per-click (CPC cost) depending on the campaign type. In particular, each advertiser has an advertiser profile (denoted as "ad profile" in the paper next) that includes ad text description, time, ad-group id, campaign id, {\it etc},   the advertiser's performance can be viewed as a mapping  from advertiser profile and bid price (such as \$2, \$3) to key performance indicators (a.k.a KPI) including number of clicks,  number of user, total amount of spend the advertiser may have, {\it i.e.,}
\begin{eqnarray}
\begin{array}{l}
\text{(ad profile, bid)} \rightarrow \text{(\#click,  \#impression,  \#spend,  \#user)}
\end{array}
\end{eqnarray}
Table~\ref{tabel:notation} summarizes the notations used in this paper. 


{\small 
\begin{table}[t]
\caption{Notation and explanations}
\label{tabel:notation}
\begin{center}
\begin{tabular}{c|c}
\hline \hline  
Notation & Explanations \\
\hline \hline
eCPM cost	& effective cost per thousand impressions \\
\hline
winning rate & the probability of winning impressions \\
\hline
CPM bid    & bid price to win per thousand impressions \\
\hline
CPC bid    & bid price to win per click \\
\hline
click(bid) & number of clicks to win if one bids at this price \\
\hline
impression(bid) & number of impressions to win if one bids at this price \\
\hline
Spend(bid) & total amount of cost if one bids at this price \\
\hline \hline
\end{tabular}
\end{center}
\end{table}
}

\begin{figure}[t]
	\centering
	\includegraphics[height=1.35in, width=3.5in]{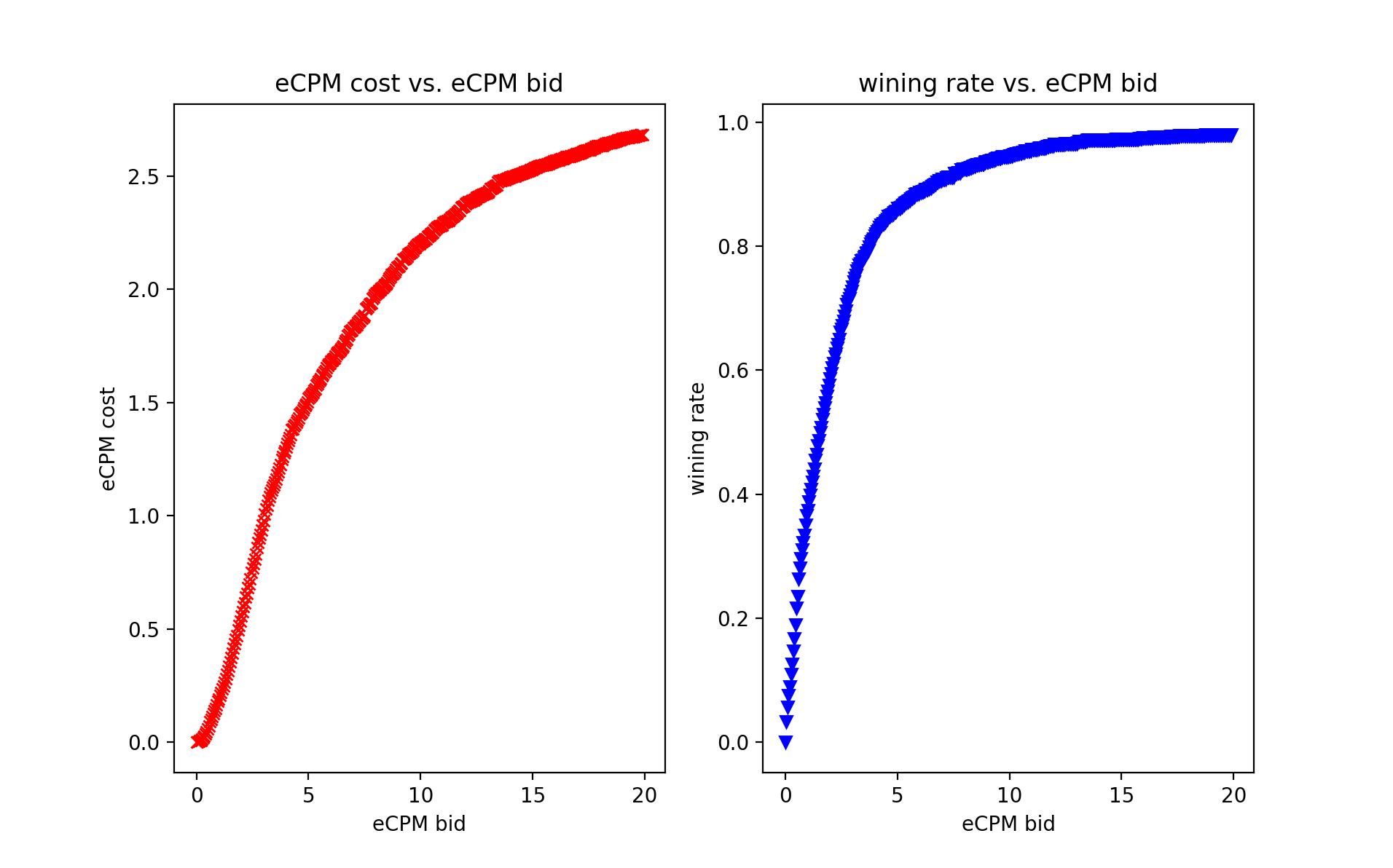} 
	\caption{\small Bid landscape model: (left) eCPM cost vs. eCPM bid; 
(right) win-rate vs. eCPM bid. }
	\label{fig:BL}
\end{figure}

\begin{figure*}[t]
	\centering
	\includegraphics[height=1.3in,width=5in]{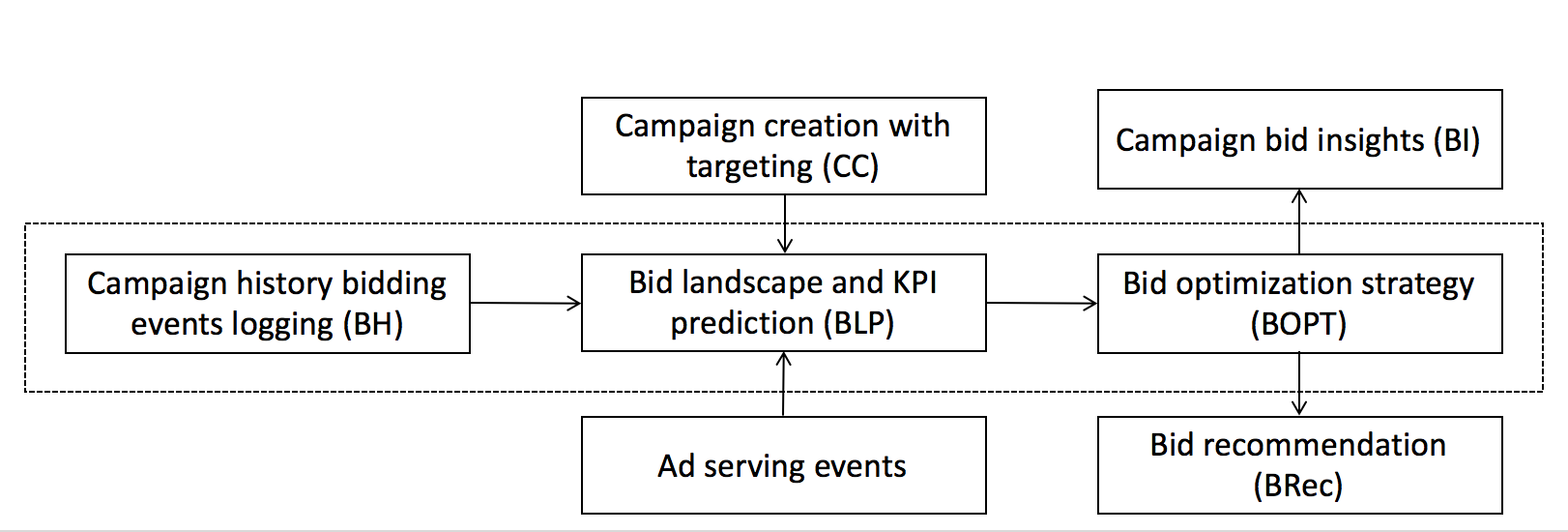}
	\caption{\small Workflow of bid optimization in advertising system.}
	\label{fig:flow}
\end{figure*}

The bid landscape model estimates the performances of the advertisers based on profiles  using history observations, which generally includes the win rate  and eCPM cost at different eCPM bid prices.  In particular, win rate reflects the percentage of winning impressions at given eCPM bid price while eCPM cost estimates the true cost at given eCPM bid price
,  which can be modeled as non-linear functions $\hat{f}$ and $\hat{g}$ respectively,  {\it i.e.,}
\begin{eqnarray}
\label{EQ:BL}
\begin{array}{l}
\label{EQ:f_winrate}
\;\;\text{Win rate} =\hat{ f} (\text{bid, ad profile}), 
\label{EQ:g_ecpm}
\;\; \text{eCPM cost} = \hat{g} (\text{bid, ad profile}).
\end{array}
\end{eqnarray}

As an example, 
Fig.~\ref{fig:BL} gives a bid landscape model generated for an advertiser profile. One important property of win rate and eCPM\_cost functions is non-decreasing monotonic property, {\it i.e.,} as the bid price increases, the true number of winning impression and  
eCPM\_cost generally increase as well. In Fig.~\ref{fig:BL}, at bidding price of eCPM \$5.00, the eCPM cost is \$1.50 and win rate is $81\%$, which indicates that the advertiser can win $81\%$ of traffic if one bids at this price. 
Win rate function $f(.)$ and eCPM cost function $g(.)$ are estimated based on generalized second-price (GSP)~\cite{NBERw11765} mechanism in auction.  After each bidder places a bid, the one with the highest bidding gets the first slot, and the second-highest wins the second slot and so on.  %
The one with the highest bidding, according to GSP, only needs to pay for the price of the second-highest bidder. The winner's eCPM cost $\hat{g}(.)$ for $i$-th position, therefore, is the successful bidder for $(i+1)$-th position\footnote{We note that there are other models used for estimating the bid landscape of advertiser, such as non-linear  
regression by gradient boosting decision trees~\cite{Cui:2011:BLF:2020408.2020454}, {\it etc}. Our model is faithful to the campaign data in observations. 
}.

Given the bid landscape of advertiser profile $j$, we will have the estimated number of clicks, the amount of cost and the amount of spend, {\it i.e.,}  
\begin{eqnarray}
\label{EQ:click_spend_cost}
\begin{array}{l}
\text{ \textcircled{a}  \;\;  Click(bid) = {\rm{Impression}} }  
\times  \text{ Winrate(bid)} \times  \text{CTR };  \\
\text{\textcircled{b}  \;\; Spend(bid)} =  \text{ Click(bid)}  \times \frac{  \text{eCPM\_{\rm{cost}}(bid)}} {1000   \times  \text{CTR}}; 
 \end{array}
\end{eqnarray}
where $\text{impression}$ is the total number of impressions the advertiser can win from the supply side (offered by publishes based on advertiser profile), CTR is the predicted click-through-rate depending on campaign, $\text{Click(bid)}$ is the number of winning clicks at price $bid$, and $\text{winrate(bid)}$ is the winning rate at price $bid$\footnote{For notation simplicity purpose, we ignore the advertiser profile index $j$ in the paper next. All discussions are based on a particular advertiser profile $j$.}.

\subsection{Use Case of Bid optimization}. Advertisers typically use different biding strategies to achieve the desired goals. 
If an advertiser is
mostly interested in getting brand name or logo in front of lots of people, then cost-per-thousand viewable impressions (vCPM) bidding strategy is a good choice since the advertiser pays for every 1,000 times the  ad appears and is viewable. 
In the typical advertising system,  advertising platform logs all campaign history auction events (BH module in Fig.\ref{fig:flow}), and then bid landscape model and key performance indicator (KPI) models (BLP module in Fig.\ref{fig:flow}) such as click-through-rate (CTR)~\cite{kongctr}, conversion rate (CVR) prediction modules are built using history observations.  After that the bid optimization strategy module (BOPT module in Fig.\ref{fig:flow}) is responsible for providing the optimized bidding strategy to advertisers. 
The bid optimization strategy plays the important roles in the following use cases (shown in Fig.\ref{fig:flow}):

$\bullet$ An advertiser wants to create a new advertising campaign and aims to find the answers to the questions such as ``\emph{what price should I bid  to achieve the best performance?}"  After running bid optimization module, we can provide insight and bid suggestions to the advertiser. 

$\bullet$ For real-time ad serving events on DSP side,  the advertiser wants to know "\emph{How much should I bid to win ad inventory on publishers?}" After running the bid optimization module, we provide the bid recommendations for advertisers to win impressions. 


\subsection{Optimization Objective}

For an advertiser, in order to maximize the returns for the amount of spend, the bid recommendation method aims to recommend the bid to meet this criterion. In particular, in the paper next, all the models are performed at \emph{advertiser profile} level\footnote{{In the paper next, ``advertiser level'' and ``advertiser profile level'' are used interchangeably, which means the same.}}. 
We aim to find the best bid price which provides the maximum click increase $d\text{Click}$ over the cost increase $d\text{Cost}$ for an advertiser, {\it i.e.,} 
\begin{eqnarray}
\label{EQ:click_cost}
\max_\text{bid}  \frac{d\text{Click(bid)}} {d\text{Cost(bid)} }, 
\end{eqnarray}
 \begin{eqnarray}
d\text{ h(bid)}  := \Delta\text{h(bid)} = \text{h}(\text{bid}+\Delta{\text{bid}}) - \text{h}(\text{bid}) ,
\end{eqnarray}
where function $h$ could be the number of click, the amount of cost, and the amount of spend, etc.  %
If we use click increase $d\text{Click}$ over spend increase $d\text{Spend}$ as the criterion, the optimal goal becomes:  
\begin{eqnarray}
\label{EQ:click_spend}
\max_\text{bid}  \frac{d\text{Click(bid)}} {d\text{Spend(bid)} }. 
\end{eqnarray}
An alternative  is using revenue increase $d\text{Revenue(bid)}$ over spend increase $d\text{Spend(bid)}$ as the criterion, 
i.e.,
\begin{eqnarray}
\label{EQ:revenue_spend}
\max_\text{bid}  \frac{d\text{Revenue}} {d\text{Spend} }.
\end{eqnarray}
In fact, the revenue depends on many factors such conversion values and numbers 
\footnote{In practice, some campaigns may not have clear conversion goals and ROI settings.}. In this paper we focus on solving Eq.(\ref{EQ:click_cost}). Eq.(\ref{EQ:click_spend}) and Eq.(\ref{EQ:revenue_spend}) can be similarly solved.

  

Now we are ready to model click increase over cost increase at advertiser level.  Here we adopt the \text{eCPM\_cost} as described above and consider the budget constraint from the advertiser. Therefore Eq.(\ref{EQ:click_cost}) becomes:
\begin{eqnarray}
\label{EQ:click_cpm_cost}
\max_\text{bid}  \frac{d\text{Click(bid)}} {d\text{ eCPM\_Cost(bid)} } \;\;\; 
s.t. \;\;\; \text{Spend(bid)} \leq \text{budget} 
\end{eqnarray}

This idea is not only applied for CPM campaigns, but also can be equivalently applied for CPC campaigns and tCPA campaigns given different ad formats and serving channels. An example for optimization of CPC campaign is: 

\begin{eqnarray}
\label{EQ:click_cpc_cost}
\max_\text{bid}  \frac{d\text{Click(bid)}} {d \text{CPC\_Cost(bid)} }  \;\;\; 
s.t. \;\;\; \text{Spend(bid)} \leq \text{budget} 
\end{eqnarray}

It is easy to see the optimal solution 
$\text{bid}^*_\text{CPC}$ of Eq.(\ref{EQ:click_cpc_cost}) can be easily obtained from the optimal solution
$\text{bid}^*_\text{eCPM}$ of Eq.(\ref{EQ:click_cpm_cost}), {\it i.e.,}
\begin{eqnarray}
\text{bid}^*_\text{CPC} = \frac{\text{bid}^*_\text{eCPM} } {1000 \times  \text{CTR}}
\end{eqnarray}

{\bf Concavity Changes}
Actually the optimization of Eq.(\ref{EQ:click_cpm_cost}) aims to find the  
a point where the concavity changes (from up to down).
In general, the tangent line to a graph at concavity change point must cross the graph at that point. Given that the sign of the slope is positive for ad bid optimization of  Eq.(\ref{EQ:click_cpm_cost}), the slopes of the tangent are decreasing as we move from left to right when the graph is concave down and increasing (from left to right) when it is concave up. 

\subsection{Variant Analysis} 

A variant of optimization goal of Eq.(\ref{EQ:click_cpm_cost}) is to consider the scale of current number of clicks and eCPM cost because the incremental increase of  
$\Delta_\text{Click(bid)}$ actually depends on the current number of clicks, i.e.,  ``normalized click vs. cost increase objective": 
\begin{eqnarray}
\label{EQ:click_cpm_cost_norm}
\max_\text{bid}  \frac{  
\frac{d\text{Click(bid)}} 
{\text{Click(bid)}}
}  
{ 
\frac{d \text{ eCPM\_Cost(bid)}}
  {\text{eCPM\_Cost(bid)}} 
} \;\;\;  
s.t. \;\;\; \text{Spend(bid)} \leq \text{budget} 
\end{eqnarray}

By considering click and spend scale, similarly, 
Eq.(\ref{EQ:click_spend})  becomes:  
\begin{eqnarray}
\label{EQ:click_spend_normalization}
\max_\text{bid} 
 \frac{
 \frac{d \text{Click(bid)}}
 {\text{Click(bid)}}
 }
  {
  \frac{ d\text{Spend(bid)} }
  { \text{Spend(bid)}}
  };   \;\;\;   
s.t. \;\;\; \text{Spend(bid)} \leq \text{budget}. 
\end{eqnarray}

\vspace{-3mm}
\begin{theorem}
Let $$ \alpha_{bid} :=  \frac{  
\frac{d \text{Click(bid)}} 
{\text{Click(bid)}}
}  
{ 
\frac{d \text{ eCPM\_Cost(bid)}}
  {\text{eCPM\_Cost(bid)}} 
};  
\\
\beta_{bid} := \frac{  
\frac{d \text{Click(bid)}} 
{\text{Click(bid)}}
}  
{ 
\frac{d \text{Spend(bid)}}
  {\text{Spend(bid)}} 
}, 
\\
\gamma_{bid} := \frac{d \text{Click(bid)}} {d \text{eCPM\_Cost(bid)} } , 
$$
we have, 

\textcircled{1}  $ \alpha_{bid} = \frac{\beta_{bid}}{1-\beta_{bid}}$, 

\textcircled{2}  $0 < \beta_{bid} < 1$,

\textcircled{3}  $\gamma_{bid} \geq  \frac{\alpha_{bid} M}{ bid } \hat{f}(bid)$, where $M$ is a constant given by $M = Impression/1000$ ( impression is determined by supplies 
and CTR is determined by the particular ad profile) and  ${\hat{f} (bid)}$ is the win rate curve as defined in Eq.(\ref{EQ:BL}). 
\end{theorem}
\vspace{-3mm}
\begin{proof}
Notice $\alpha_{bid}, \beta_{bid}$,  $\gamma_{bid}$, $spend(bid)$,  $eCPM\_cost(bid)$,  $Click(bid)$ are all functions of $bid$, for notation simplicity, we omit bid in 
the following proof. 

\textcircled{1} Notice that
\begin{eqnarray}
\label{EQ:spend}
Spend = \frac{eCPM\_cost}{1000 \times CTR} \times Click, 
\end{eqnarray}
 \begin{eqnarray}
\label{EQ:dspend}
dSpend = L \times \Big( d eCPM\_cost \times Click  +  dClick \times eCPM\_cost \Big),
\end{eqnarray}
where $L =1/ (1000\times CTR)$.
Substitute $Spend$ of Eq.(\ref{EQ:spend}) and $dSpend$ of Eq.(\ref{EQ:dspend}) into $\beta_{bid}$, we can get 
\begin{eqnarray}
\label{EQ:beta_bid=}
\beta_{bid} &:=&  \frac{  
\frac{d\text{Click}} 
{\text{Click}}
}  
{ 
\frac{d\text{Spend}}
  {\text{Spend}} 
}  =  \frac{d{\text{Click}}}  { d\text{Spend}}  \frac{\text{Spend}}{\text{Click}} \nonumber \\
& =&   \frac{d \text{click}  \times \text{eCPM\_cost} } { d \text{eCPM\_cost}  \times \text{Click} + d \text{Click}  \times \text{eCPM\_cost}} \nonumber \\
& =&  \frac {  \frac{d\text{Click}} {d \text{eCPM\_cost}}  \frac{\text{eCPM\_cost}}{\text{Click}} 
    } 
{1 +   \frac{d\text{Click}} {d \text{eCPM\_cost}}    \frac{\text{eCPM\_cost}}{\text{Click}}    } \nonumber \\
&=& \frac{\alpha_{bid}}{ 1 + \alpha_{bid}}. 
\end{eqnarray}

\textcircled{2}  Given the fact that the higher bid always gives the larger win rate and the higher eCPM\_cost in bid landscapes,  from Eq.(\ref{EQ:click_spend_cost}), one can tell $Click(bid)$ and $Spend(bid)$ are both monotonically non-decreasing functions with respect to bid due to the non-decreasing property of win rate and eCPM\_cost functions.  Therefore,  
$d\text{Spend} \geq 0 $ and  $d \text{Click} \geq 0$,  where the equality holds when the current spend reaches to the budget. Even if one increases the bid, the spend and click will not increase anymore.  Therefore, $\beta_{bid} > 0$.  From Eq.(\ref{EQ:beta_bid=}), one can tell $\beta_{bid} < 1$.  

\textcircled{3} Notice that
\begin{eqnarray}
\label{EQ:click+}
 Click = Impression \times CTR \times win rate,  \;\; 
  \frac{CPC\_cost} {bid} \leq 1. \nonumber
\end{eqnarray}
We have,  
\begin{eqnarray}
\label{EQ:gamma_bid}
\gamma_{bid} &:=&  \alpha_{bid} \times \frac{\text{Click}}{\text{eCPM\_cost}}  \nonumber \\
&=&  \alpha_{bid} \frac{ Impression \times CTR \times win rate } { 1000 \times  \text{CPC\_cost} \times \text{CTR} } \nonumber \\
&=& \alpha_{bid} \frac{M \times win rate}{CPC\_cost}  \nonumber \\
&\geq&   \frac{\alpha_{bid} \times M}{bid}  \hat{f}(bid).
\end{eqnarray}
This completes the proof. 
\end{proof}

{\bf Discussion} 
The major drawback of $\alpha_{bid}$ and $\beta_{bid}$ is the ignorance of the impact of absolute numbers of click and cost at current bid by using the normalized ratio. In practice, we are not only interested in the normalized ratio of click increase vs. cost increase, but more interested in the true numbers of clicks and eCPM\_cost that are viewed as important business metrics (such as revenue-per-mile (RPM), click yield (CY), {\it etc}).  Therefore, in the paper next, we use Eq.(\ref{EQ:click_cpm_cost}) as the optimization goal.

\vspace{-2mm}
\section{Concavity Changes Detection Model}
\label{sec:LGIM}

At first sight, it seems very easy to solve Eq.(\ref{EQ:click_cpm_cost}), since  we can do constrained optimization based on the estimations from Eq.(\ref{EQ:click_spend_cost}).  However, there are several major challenges we need to address in order to solve the optimization problem of Eq.(\ref{EQ:click_cpm_cost}).
Firstly, solving Eq.(\ref{EQ:click_cpm_cost})  is a discrete constraint optimization problem because click and spend functions are distributed over a range of bid prices unevenly (challenge C1). 
Secondly, solving Eq.(\ref{EQ:click_cpm_cost})  depends on the estimation of clicks and spend in  Eq.(\ref{EQ:click_spend_cost}). However, the estimations 
would be noisy and have missing values since eCPM cost distribution and win rate distributions are sparse and noisy (challenge C2). 
Finally, the proposed algorithm is required to be fast, accurate and robust with minor overhead (challenge C3). 

We propose the logistic growth inflection model to address the above challenges. %
In particular, we transform the discrete optimization problem into a continuous optimization problem by interpolating the unseen cost prices using a concave sigmoid type function at different cost prices (for challenge C1). We use the non-linear regression to fit the continuous curves with the purpose of removing the noises and missing values automatically by capturing the overall click distributions over the cost (for challenge C2).  We propose an effective binary search method to solve the constrained optimization problem with guaranteed optimal solution. Moreover, the estimated spend and click distributions are done offline as the pre-processing step (for challenge C3).   Essentially, the proposed logistic growth inflection model consists of three key steps which will be illustrated in detail in the following subsections. 

$\bullet$ Model $\frac{\Delta_\text{Click}} {\Delta_\text{eCPM\_Cost} }$  using logistic growth sigmoid function;

$\bullet$ Optimize model parameters with non-linear regression; 
 
$\bullet$ Solve the resultant constraint optimization problem using binary search in continuous domain.

\vspace{-2mm}
\subsection{Function Approximation}

Notice that number of $\text{Click}$ and $\text{eCPM\_Cost}$ are both functions of bid price, and therefore, $\text{Click}$ can be modeled as a function of $\text{eCPM\_Cost}$, {\it i.e.,}
\begin{eqnarray} 
\label{EQ:h}
\text{Click}  \triangleq  h(\text{eCPM\_Cost}),
\end{eqnarray}
where $h$ is a function that maps the $\text{eCPM\_Cost}$ to the number of $\text{Click}$.  Then we have, 
\begin{eqnarray}
\label{EQ:del}
  \frac{\Delta_\text{Click}} {\Delta_\text{eCPM\_Cost} }  = \frac{dh} {d \text{eCPM\_Cost}} = h'(\text{eCPM\_Cost}), 
\end{eqnarray}
where $h'$ is the first-order derivative of function $h(.)$.  Therefore Eq.(\ref{EQ:click_cpm_cost}) is transformed to solving the following problem, 
 \begin{eqnarray}
\label{EQ:del2}
& \max_\text{eCPM\_cost}   h'(\text{eCPM\_Cost}),  \nonumber \\
& s.t. \;\;\; \text{Spend(eCPM\_cost)} \leq \text{budget},
\end{eqnarray}
and its advantages over Eq.(\ref{EQ:click_cpm_cost}) can be clearly seen in the following discussions. 

\vspace{-2mm}
\subsubsection{Function requirement}

More importantly, there are several nice properties that function $h(.)$ of Eq.(\ref{EQ:h}) is required to have:  

\textcircled{1} If the cost is zero, then the click is zero.

\textcircled{2} If the cost is very high, then the click should reach to an upper limit (a.k.a asymptote mathematically).    

\textcircled{3} The function is monotonically increasing,  {\it i.e.,} the number of clicks at higher cost is always greater that at  lower  cost. Mathematically, $h(x+\epsilon) > h(x)$ for $\epsilon >0$.

\textcircled{4} If the cost is between zero and the upper limit, an S-curve function is desired due to the following reasons.  Generally, at lower cost, at first some minimum amount of cost increase will generate noticeable output click creation. Once the market passes the inflection point at which a change in the direction of curvature occurs, modest increases in the cost input generate disproportionate movements in output clicks. In other words, before inflection point, very small increases in the input cost generate modest increases in the output clicks. After inflection point, very large increases in the input cost generate only modest increases in the output before the "practical ceiling" where at some point, additional increases in the input cost generate no noticeable click increase at all.


\vspace{-2mm}
\subsubsection{Modeling using sigmoid-type function}

In order to satisfy the above properties, we use logistic growth sigmoid-type function to model $h(.)$.  For notation simplicity, let
$$x = \text{eCPM\_cost},  \;\;  h(x) = click,$$
\begin{eqnarray}
\label{EQ:h(x,y)}
 h(x) = \frac{s}{1+\exp^{-t x+p}} - q 
\end{eqnarray}
where $s, t, p, q \ge 0 $ are model parameters of $h(.)$ and 
$\frac{s}{1+\exp^p} - q = 0\;\; (s -q > 0). $
Essentially, this function is transformed based on standard sigmoid function 
$f(x) = \frac{1}{1+exp^{-wx}}$, which moves down by $\frac{q}{s}$ in y-direction, and moves right by $\frac{p}{t}$ in x-direction and stretches 
by $s$ in y-direction. Therefore, it maintains very similar functionality as standard sigmoid function, including S-curve type derivatives. Notice that 
\begin{eqnarray}
\label{EQ:h=0}
h(0) &=& 0, \;\;  \nonumber \\ 
\lim_{x \rightarrow \infty} h(x) &=& s - q,   \;\;  \nonumber \\
\frac{d h}{dx} &=& \frac{t  s \cdot \exp^{-tx+p}}{(1+\exp^{-tx +p})^2} \ge 0, 
\end{eqnarray}
therefore $h(x+\epsilon) > h(x)$ and properties (1--3) are all satisfied.  Further we have Lemma~\ref{lemma:s-curve} that states the S-curve property of $h(.)$. 

\subsubsection{Why not other functions?} We note that there are alternatives for modeling the relations between Click and eCPM\_cost, most of which, however, do not exhibit the S-curve property. 
For example, if we use \emph{power function} to model it, {\it i.e,}
$y = \alpha x^{\beta}, \;\; 0 < \beta \leq 1.$
This function will never saturate since $x \rightarrow \infty $, when $x \rightarrow \infty $. Another alternative is \emph{Michaelis-Menten function}, {\it i.e.,}
$y = \frac{\alpha x}{1+\beta  x}, \;\; \beta > 0 .$
Although this function has the added bonus of reaching the click saturation of $\frac{\beta}{\alpha}$, it does not exhibit the S-curve property correctly. Furthermore, 
if we use \emph{negative exponential families}, {\it i.e.,}
$y = \alpha(1 - \exp^{-\beta  x}),  \;\; \beta >0,$
the maximum click is only attained at saturation point $\alpha$ and there is no S-curve type property as well.  
From the above analysis, it is clearly that $h(.)$ of Eq.(\ref{EQ:h(x,y)}) is the best candidate for modeling the logistic growth relations, which is also further validated by the following Theorem~\ref{lemma:s-curve}.

 \begin{theorem}
\label{lemma:s-curve}
The function $h(.)$  of Eq.(\ref{EQ:h(x,y)}) has the S-type curve, which increases faster and faster before reaching to the peak values, and then increases slower and slower after that.  
\end{theorem} 
\begin{proof}
This is verified by 
2nd-order derivative $\frac{d^2 h}{dx^2}$, {\it i.e.,}
$\frac{d^2 h}{dx^2} > 0 $ when $0 \leq x < \frac{1}{t}log(2 - \exp^p)$ , 
$\frac{d^2 h}{dx^2} = 0 $ when $ x =  \frac{1}{t}log(2 - \exp^p)$,  
and $\frac{d^2 h}{dx^2} < 0 $ when $ x > \frac{1}{t}log(2 - \exp^p)$.
This is easy to prove by setting $\frac{d^2 h}{dx^2} = 0$, {\it i.e.,}
\begin{eqnarray}
\frac{-t^2 (1 + \exp^{-tx + p} - 2 \exp^{-tx}) s \exp^{-tx}}{ (1+\exp^{-tx+p})^3}  = 0, 
\end{eqnarray}
and this gives the peak value,
 $x* = \frac{1}{t}log(2-\exp^p).$  
One can check when $ 0 \leq x < x^*$,  $\frac{d^2 h}{dx^2} > 0 $.
When $x >x^*$,  $\frac{d^2 h}{dx^2} < 0 $.
\end{proof}

Eq.(\ref{EQ:h(x,y)}) gives a very simple but effective solution to model the relations of click vs. eCPM\_cost given the observations. Moreover, this function enjoys the nice properties of the S-curve function. In particular, we use the logistic growth type function, which is the best-known example of S-curve behaviors.  Many growth processes, such as diffusion of innovations, population growth, {\it etc} exhibit varying speeds of exponential growth at first, then they hit an inflection point after maturing, running into competition (or scarcity), and finally the rate of growth decelerates. This exactly fits the modeling of clicks vs. eCPM\_cost in ad context.


\subsection{Parametric Learning Method} 
Concerning fitting the model parameter in Eq.(\ref{EQ:h(x,y)}), we have observed the true clicks ($y_j$) and eCPM\_cost ($x_j$) in different auctions $j$, {\it i.e.,} $\{x_j, y_j\}_{j=1}^{n}$. These observations are modeled 
using model parameters $\theta = [s, t, p, q]$ in Eq.(\ref{EQ:h(x,y)}). Generally, there is no closed-form solution for the best-fitting parameters given this non-linear regression problem, and we fit the model by approximating the observations from refining the parameters in successive iterations.  In particular, 
using the least square loss function, 
the sum of square errors should be minimized, {\it i.e.,}
\begin{eqnarray}
\label{EQ:S}
\min_{\theta} S = \min_{\theta} \sum_{j=1}^n r^2_j, \;\; r_j = h(x_j) - y_j,
\end{eqnarray}
where $r_j$ is the error term.
Then the minimum value of $S$ is achieved when the gradient $\frac{\partial S}{\partial \theta} = 0$, {\it i.e.,}
$
\frac{\partial S}{\partial \theta_s} = 2 \sum_j r_j \frac{\partial r_j}{ \partial \theta_s} = 0,
$
which is actually a function of \text{eCPM\_cost} and parameters $s, t, p, q$.  Usually, there is no closed-form solution.
After giving the initial guess, the parameters $\theta$ are refined iteratively in successive approximations. Let Jacobian matrix $J =[J_{jk}]
$ and $J \in \Re^{n \times K}$, $\frac{\partial r_j}{\theta_s} = -J_{js}$,  then
\begin{eqnarray}
\label{EQ:rj}
 r_j &=& y_j  - h(x_j; \theta)  \nonumber \\
 &=& \Big(y_j - h(x_j; \theta^t) \Big) +  \Big(h(x_j; \theta^t) - h(x_j; \theta)  \Big) \nonumber \\
&=&  \Delta {y_j} - \sum_{s=1}^n J_{js} \Delta \theta_s
\end{eqnarray}
Then we have
\begin{eqnarray}
\label{EQ:J_ijs}
\sum_{j=1}^n \sum_{k=1}^K J_{jk} J_{js}  \Delta {\theta_k} = \sum_{j=1}^n J_{js} \Delta y_j.  
\end{eqnarray}
This gives the basic updating equations used in Gauss-Newton method in matrix form, {\it i.e.,}
$(J^T J) \Delta \theta = J^T \Delta y.  $
The algorithm stops when the stop criterion is satisfied, {\it i.e.,}
$\frac{|S^{t+1} - S^t|}{|S^t|}  \leq \xi$,
which is set to \text{1e-5} in practice.

     


Now we are ready to solve Eq.(\ref{EQ:del2}) given the fitted model parameters  $\theta = [s, t, p, q]$ .  Notice that the optimal $eCPM\_cost$ is given by: 
\begin{eqnarray}
\label{EQ:ecpm_cost**}
\text{eCPM\_cost*} ={argmax}_x h'(x).  \;\;\;
 \end{eqnarray}
As is discussed before, we can obtain the optimal solution $x^*$
by setting the second order of derive to zero, {\it i.e,} 
\begin{eqnarray}
 \frac{d^2 h}{dx^2} = 0, \; \;\; 
  x^*  =  \frac{1}{t}log(2 - \exp^p).
  \label{EQ:x*}
 \end{eqnarray}

This is not done. We need to further check whether the spend exceeds the budget or not. In particular, if 
$
Spend(\text{eCPM\_cost*} ) \leq budget, 
$
then  $\text{eCPM\_cost*} = x^*$ is the optimal solution.  Otherwise, if the spend already exceeds the budget, we need to decrease \text{eCPM\_cost*} until spend is within the budget because $h(.)$ is a non-decreasing function.  In particular, we use a binary search algorithm to search for the highest \text{eCPM\_cost*} with the amount of spend at most equal to the budget, {\it i.e.,}
\begin{eqnarray}
\label{EQ:ecpmbudget}
\text{eCPM\_cost*} = \text{binary\_search}(0, x^*, \text{budget}).
\end{eqnarray}
In more detail, Algorithm ~\ref{alg:bst} summarizes the key steps in binary search to find the best cost price.  

\begin{algorithm}
\label{alg:bst}
 \caption{\small Binary search algorithm to find the maximum ecpm\_cost (Eq.\ref{EQ:ecpm_cost**}) given fixed budget. }
\begin{flushleft}
 {\bf Input:}   $x^*$, $budget$,   \\
{\bf Output:} $\text{eCPM\_cost}^*$   \\
\end{flushleft}
\begin{algorithmic} [1]
\STATE { min\_cost   $\leftarrow$ 0, max\_cost $\leftarrow$ x*. }
  \WHILE{min\_cost $\leq$  max\_cost}
      \STATE { mid\_cost $\leftarrow$  (min\_cost + max\_cost)/2 }.
      \IF  {spend(mid\_cost) $\approx$ budget according to Eq.(\ref{EQ:click_spend_cost})}
      	  \STATE {$ecpm\_cost^* \leftarrow mid\_cost$; }
      	  \RETURN   \text{ecpm\_cost*} 
      \ELSIF {spend(mid\_cost) $>$ budget}
           \STATE {max\_cost $\leftarrow$ mid\_cost-0.001; }
      \ELSE
            \STATE {min\_cost $\leftarrow$ mid\_cost+0.001; }
      \ENDIF
    \ENDWHILE
\end{algorithmic}
\label{alg:bst}
\end{algorithm}


\begin{figure*}
	\centering
	\begin{subfigure}[b]{0.33\textwidth}
	\includegraphics[height=1.2in, width=\linewidth]{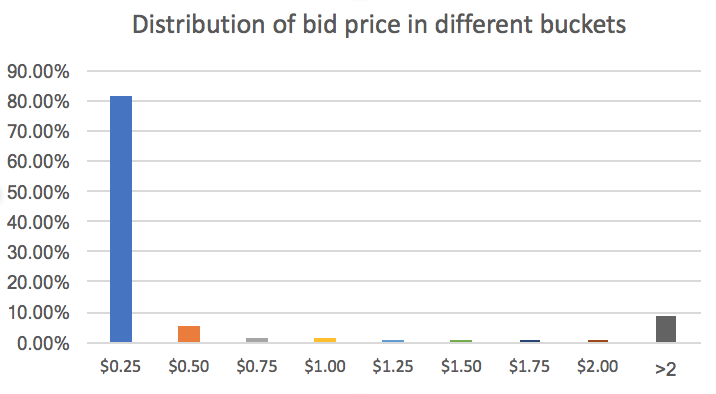}
	\caption{Bid distribution}
		\label{fig:bid}
	\end{subfigure}
	\begin{subfigure}[b]{0.33\textwidth}
	\centering
	\includegraphics[height=1.2in, width=\linewidth]{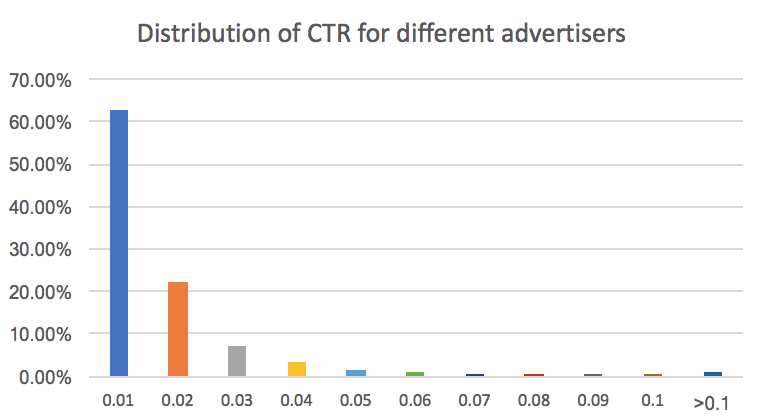}
		\caption{CTR distribution}
	\label{fig:ctr}
	\end{subfigure}
	\begin{subfigure}[b]{0.33\textwidth}		
	\centering
	\includegraphics[height=1.2in, width=\linewidth]{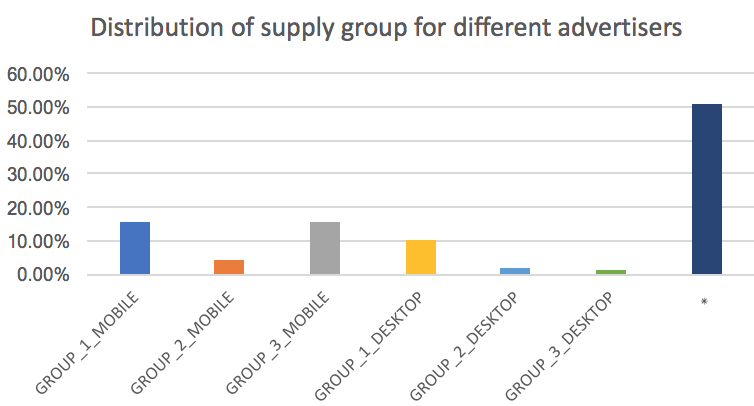}
	\caption{Supply group distribution}
	\label{fig:supply}
	\end{subfigure}		
	\caption{\small (a) distribution of bid price at different buckets; (b) distribution of CTR for advertisers at different buckets; (c) distribution of supply group 
	at different buckets.}
	\label{fig:insight}
\end{figure*}

\begin{figure*}
	\centering
	\begin{subfigure}[b]{0.33\textwidth}
	\includegraphics[height=1.5in, width=\linewidth]{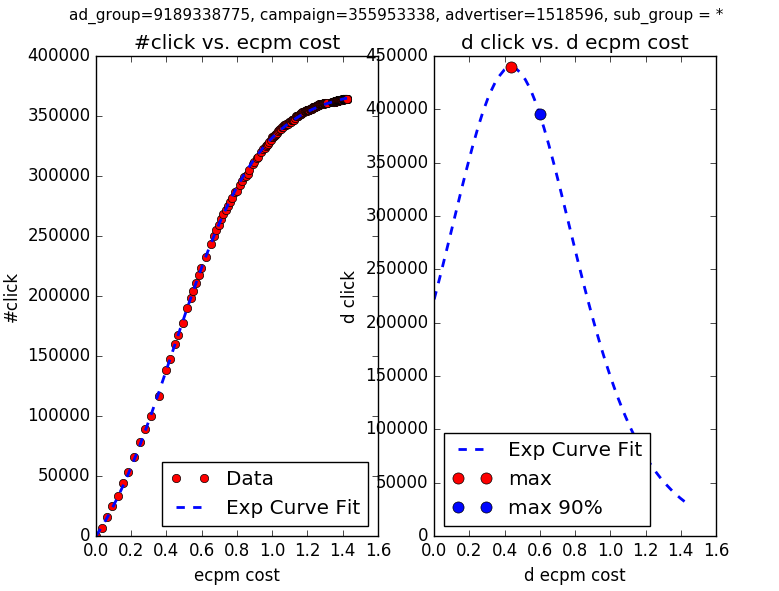}
	\caption{campaign 1}
		\label{fig:native-1}
	\end{subfigure}
	\begin{subfigure}[b]{0.33\textwidth}		
	\centering
	\includegraphics[height=1.5in, width=\linewidth]{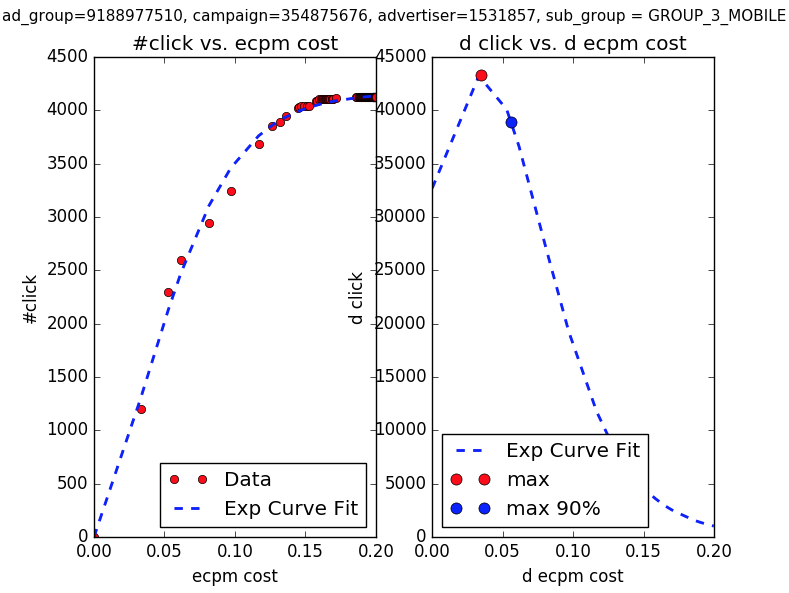}
	\caption{campaign 2}
	\label{fig:native-4}
	\end{subfigure}	
	\begin{subfigure}[b]{0.33\textwidth}
	\centering
	\includegraphics[height=1.5in, width=\linewidth]{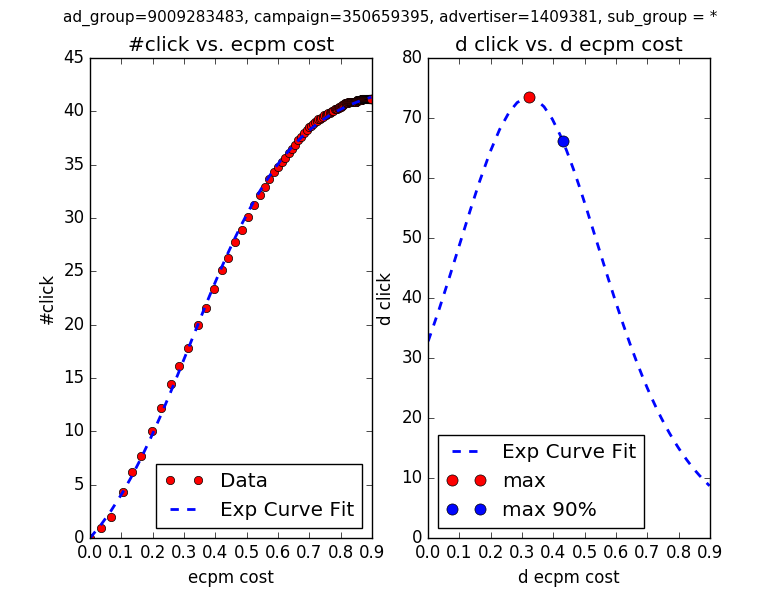}
		\caption{campaign 3}
	\label{fig:native-2}
	\end{subfigure}
	\begin{subfigure}[b]{0.33\textwidth}		
	\centering
	\includegraphics[height=1.5in, width=\linewidth]{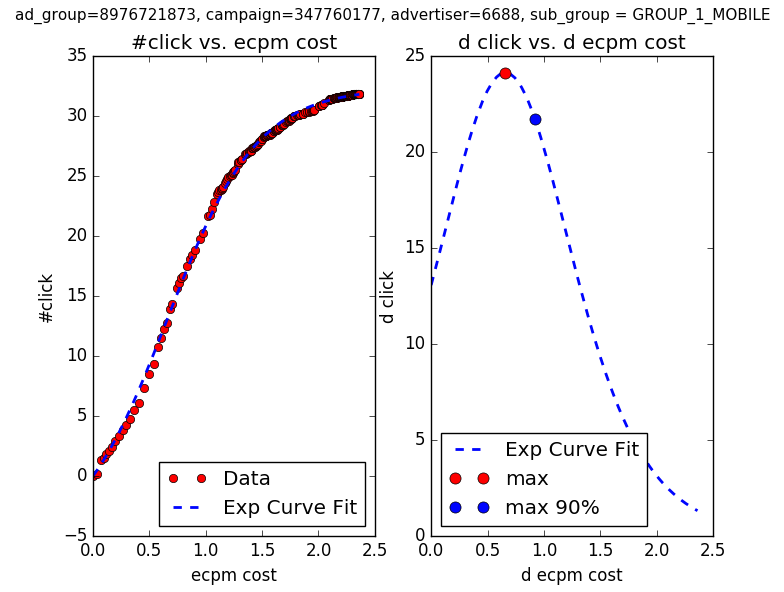}
	\caption{campaign 4}
	\label{fig:native-3}
	\end{subfigure}		
	\begin{subfigure}[b]{0.33\textwidth}
	\includegraphics[height=1.5in, width=\linewidth]{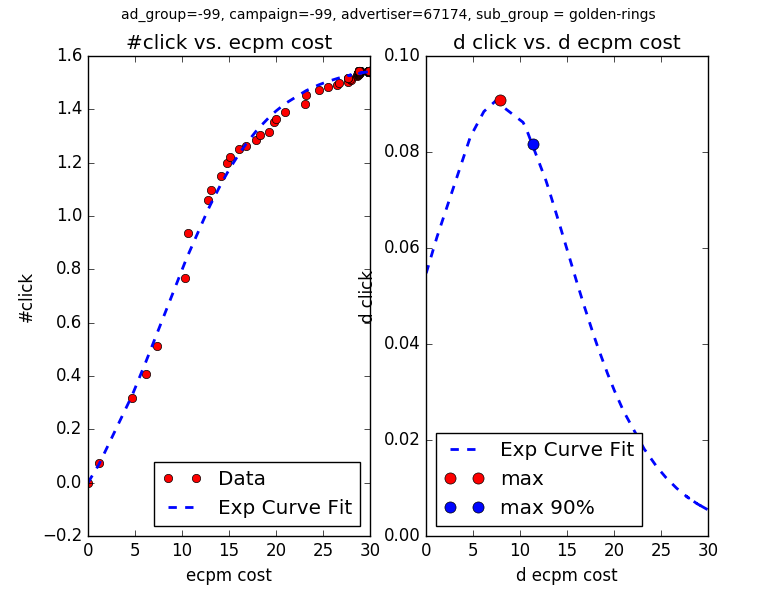}
	\caption{campaign 5}
		\label{fig:search-1}
	\end{subfigure}
	\begin{subfigure}[b]{0.33\textwidth}
	\centering
	\includegraphics[height=1.5in, width=\linewidth]{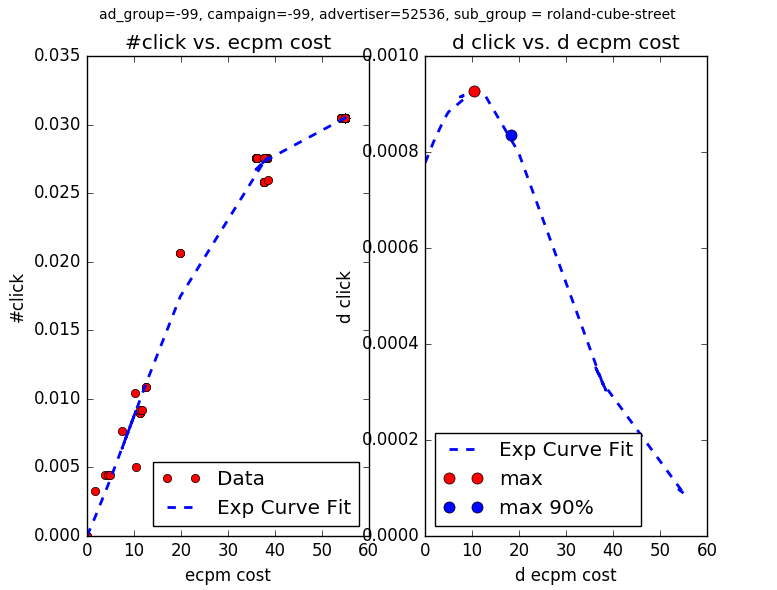}
		\caption{campaign 6}
	\label{fig:search-2}
	\end{subfigure}
	\caption{\small Number of Click vs. eCPM\_cost curves and the corresponding derivatives ($\frac{d \text{Click}}{d \text{eCPM\_cost}}$) on six campaigns.  In left panel of each subfigure,  we use the S-curve function (in dashed blue) to fit the observations (in red circle). In right panel of each figure, the maximum derivate (i.e., the inflection point) is obtained in red circle in derivative curve (in dashed blue). "Max 90\%" is the point that achieves 90\% of maximum derivatives but goes beyond the inflective point.}
	\label{fig:native-search-3}
\end{figure*}

\section{Evaluations}
\label{sec:exp}



\subsection{Offline Evaluation}

{\bf Adverting campaign logs} We collect advertising campaign log data on a major advertising platform in
U.S after removing inactive campaigns. We take a sample of 15\% of all bidding records in one-week, totally 135,409 native ad campaign 
profiles in offline evaluation.  All these campaigns focus on clicks for displaying advertisement that well matches the form and function of the publisher platform upon which it appears.
The campaigns use cost-per-click (CPC) bidding,  and the advertiser needs to pay only when someone actually clicks on the ad after visiting publisher site. In  auctions the advertising campaigns win the participated inventory based on GSP.  
The campaigns can win impressions from 7 supply groups, including desktop and mobile platforms with 3 tiers of 
impression supplies denoting different publisher groups, and one unified supply group for any platform of any publisher. We show the distribution of bid price,
 CTR of advertisers, and supply groups in Fig.\ref{fig:insight}.

{\bf Correctness of solution}
 Fig.\ref{fig:native-search-3} plots six examples of number of Click vs. eCPM\_cost curves and the corresponding derivatives ($\frac{d \text{Click}}{d \text{eCPM\_cost}}$) using the advertising campaign dataset. One can easily find the inflection point using the proposed logistic growth function no matter whether there exist noisy or missing values in observations. %
 Essentially, the outliers and missing values are smoothed out by the logistic growth function, making the result more robust and reliable.   A naive method is to discover the inflection point from directly computing the gradient for each observation before obtaining the point with the largest derivative of click increase over eCPM\_cost increase. However, the major drawback of this method is  that one may miss the true optimal solution due to the noisy and sparse observations.  Fig.\ref{fig:fail-case} shows an example of failure case, where it is hard to tell the optimal bid from the noisy first-order derivatives. %
Let the $cost\_{naive}$ be the optimal eCPM\_cost using naive method, and let $cost\_{cf}$ be the optimal eCPM\_cost of our method. We define the {\bf derivative increase ratio} (DiffR) measurement to quantify the difference of them, {\it i.e.,}
\begin{eqnarray}
DiffR =\frac{ \frac{d \text{Click}}{d \text{eCPM\_cost}}|_{cost\_naive}  - \frac{d \text{Click}}{d \text{eCPM\_cost}}|_{cost\_{cf}}  }
{
\frac{d \text{Click}}{d \text{eCPM\_cost}}|_{cost\_naive} .
}
\end{eqnarray}
The mean of DiffR is found to be +24.56\% by performing experiments on our advertising campaign dataset. 



{\small
\begin{table}
\small
\caption{\small MAPE, RMSE of predicted \emph{click} using 6 methods: (1) NNS, (2)  LI,
(3) Power, (4) MM,   (5) Exp, (6) Sigmoid. The smaller, the better.
}
\begin{center}
\begin{tabular}{ c|c|c}
\hline
\hline
Methods & MAPE & RMSE \\
\hline
Nearest Neighbor Search (NNS) & 0.2619  & 0.2931 \\
 Linear interpolation (LI) &   0.2378 & 0.2547\\
 Power function (Power) & 0.2406  & 0.2482 \\
Michaelis-Menten (MM)&   0.2279 & 0.2503\\
Negative exponential (Exp) & 0.2028  & 0.2380 \\
Sigmoid (proposed) &   0.1303  & 0.1572\\
\hline
\hline
\end{tabular}
\end{center}
\label{tbl:click-error}
\end{table}
}
{\small
\begin{table}
\small
\caption{\small MAPE, RMSE of predicted \emph{spend} using 6 methods: (1) NNS, (2)  LI,
(3) Power, (4) MM,   (5) Exp, (6) Sigmoid.
}
\begin{center}
\begin{tabular}{ c|c|c}
\hline
\hline
Methods & MAPE & RMSE \\
\hline
Nearest Neighbor Search (NNS) & 0.2342  & 0.2749 \\
 Linear interpolation (LI) &   0.2232 & 0.2565\\
 Power function (Power) & 0.2096  & 0.2134 \\
Michaelis-Menten (MM)&   0.1980 & 0.2012\\
Negative exponential (Exp) & 0.1797  & 0.1978 \\
Sigmoid (proposed) &   0.1269 & 0.1634\\
\hline
\hline
\end{tabular}
\end{center}
\label{tbl:spend-error}
\end{table}
}
\begin{figure}
	\centering
	\includegraphics[height=1.0in, width=0.5\linewidth]{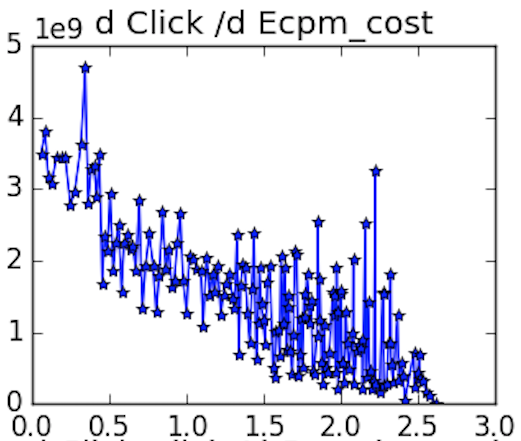}
	\caption{\small failure case using naive method for bid recommendation}
	\label{fig:fail-case}
\end{figure}

\subsubsection{Sigmoid curve fitting baseline comparisons}


In order to evaluate how accurate the fitted curves, we pull out the current true bid price, true clicks and spend\footnote{The amount of spend is computed using the following equation:
$Spend = Click \times eCPM\_cost \times \frac{1}{1000 \times CTR}.$
},  and then compare them against the predicted clicks and spend at the same bid price using the proposed sigmoid-type curve fitting technology.  We adopt the mean of absolute percentage error (MAPE) and root mean square error (RMSE)  as the measurement, {\it i.e.,}  
{\small
\begin{eqnarray}
MAPE =  \sum_{i=1}^n \frac{1}{n}\frac { |{y}_i-  \hat{y_i}| }{y_i},    \;\;\; 
RMSE = \sqrt{\frac{1}{n}\sum_{i=1}^n{(\hat{y_i}-y_i)}^2},
\label{EQ:ape}
\end{eqnarray}
}
where $y_i$ is the true click or true spend, and $\hat{y_i}$ is the forecasted number. The smaller, the better. 
The baseline methods we compared include:

$\bullet$ {\bf Nearest Neighbor Search (NNS).} 
For the observed bid $b$,  if $b'_i$ is the nearest neighbor of bid $b_i$ in history, then we retrieve the corresponding click and spend for bid $b'_i$ as the predicted values .  
 
$\bullet$  {\bf Linear interpolation (LI)} 
The predicted $\hat{click(b)}$ at bid $b$ is the linear interpolation based on $b_i$ and $b_j$, {\it i.e.,}
$\hat{click(b)}= click(b_i) + (b- b_i)\frac{click(b_j) - click(b_i)}{b_j - b_i}, $
where $click(b_i)$ is the observed click at bid $b_i$ and $click(b_j)$ is the observed click at bid $b_j$.  The amount of spend can be similarly computed.

$\bullet$ {\bf Power function} The function\footnote{\url{http://wmueller.com/precalculus/families/1\_41.html}} is $y = \alpha x^{\beta}, \;\; 0 < \beta \leq 1,$  where $x$ denotes eCPM\_cost, $y$ denotes number of clicks. 

$\bullet$ {\bf Michaelis-Menten function}  The function\footnote{\url{https://en.wikipedia.org/wiki/Michaelis\%E2\%80\%93Menten\_kinetics}} is $y = \frac{\alpha x}{1+\beta  x}, \;\; \beta > 0 ,$  where $x$ denotes eCPM\_cost, $y$ denotes number of clicks. 

$\bullet$ {\bf Negative exponential function} The function\footnote{\url{https://en.wikipedia.org/wiki/Exponential\_distribution}} is $y = \alpha(1 - \exp^{-\beta  x}),  \;\; \beta >0 ,$  where $x$ denotes eCPM\_cost, $y$ denotes number of clicks.  

The parameters $\alpha, \beta$ for the power function, Michaelis-Menten function and negative exponential function are learned from observations by minimizing the least square errors as is shown in \S 3.2.   Tables~\ref{tbl:click-error}, ~\ref{tbl:spend-error} show the performance comparisons agains these baselines. Clearly, the proposed sigmoid function gives better results in terms of both MAE and RSME. 
As in other parametric learning method, the best fitted parameters depend on auction and campaigns. The superiority of the inflection model is reflected at aggregation-level  by considering different campaigns and auctions. 

\subsubsection{Bid strategy comparisons}

We compared four different bid optimization strategies:

$\bullet$ Strategy 0 (no-opt): is the current bid without any bid optimization. 

$\bullet$ Strategy 1 (mc): is the bid optimization strategy that achieves the maximum number of clicks from the advertisers' view,  and the optimal bid is 
$bid_{MC}^* = argmax_\text{bid}  {\text{Click(bid)}}$ within the budget constraint 
$ \text{Spend(bid)} \leq \text{budget} .$

$\bullet$ Strategy 2 (mc90): is the strategy that optimal bid $bid_{MC90}^*$  achieves 90\% of maximum clicks. 

$\bullet$ Strategy 3 (ip): is the proposed inflection point approach defined in Eq.(\ref{EQ:click_cpm_cost}); 

$\bullet$ Strategy 4 (ip90): is the strategy that optimal bid  $bid_{IP90}^*$ achieves 90\% of  $\frac{d \text{Click}}{d \text{eCPM\_cost}}$ but goes beyond the inflection point. %

$\bullet$ Strategy 5 (cpa): is the strategy that optimal bid  $bid_{cpa}^*$ achieves the maximum number of conversions.  

$\bullet$ Strategy 6 (roi): is the strategy that optimal bid  $bid_{roi}^*$ achieves the maximum values of returns on investment~\cite{Perlich:2012:BOI:2339530.2339655}. 
Notice that oCPC strategy~\cite{Zhu:2017:OCP:3097983.3098134} is considered for increasing gross merchandise volume (GMV) in Taobao platform. However, different from e-commerce website, in our advertising platform for content media, the optimization of GMV cannot be applied.

{\bf Measurement}  We introduce several measurement to evaluate the performance of different bid optimization strategies.  RPM is indicator for advertising revenue for per thousand impressions, PPC is the indicator for pay per click for advertisers, ROI is to measure advertisers' return on investment. We also define Click Yield Ratio (CYR) as CYR = $\frac{d \text{Click}}{d \text{eCPM\_cost}}$ which measures the click increase over the spend.  We show the comparison results using different strategies against the baseline (strategy 0) as follows. 

\begin{tabular}{ p{1.5cm}|c|c|c|c}
\hline
\hline
Strategy & RPM & PPC & ROI & CYR  \\
\hline
1 (mc) & 29.17\%  & 32.16\%  &10.76\% & 18.67\%\\
2 (mc90) &   24.90\% & 26.45\% & 12.12\%&  16.25\%\\
3 (ip) & 23.57\% & 24.16\% & 16.34\%& 26.43\%\\
4 (ip90) &   20.19\% & 19.46\% &14.36\% & 24.58\%\\
5 (cpa) & 27.46\%  & 29.78\%  &11.13\%& 17.36\%\\
6 (roi) &   1.78\%  & 2.34\% &20.15\%& 3.21\% \\
\hline
\hline
\end{tabular}

We observe that strategy 6 (roi) focuses on optimizing advertiser's ROI and cannot ensure better RPM and PPC, which encourages the advertisers to bid at a much lower price to achieve the best ROI leading to insufficient click increase without spending the budget. Strategy 1, 2 and 5 encourage advertisers to spend as much as they can, in order to achieve the maximum number of clicks (or conversions), which generally require to increase the bidding approximately 20\% but end with limited ROI increase.   Proposed inflection point strategy (i.e., strategy 3 and 4) stands in between ROI strategy and maximum click (or conversion) strategy, which encourages the increase of bidding price to inflection point within the budget,  leading to much more ROI and click yield increase.

{\bf Campaign results} Besides the overall performance, we also show the result of particular campaigns (demonstrated in Fig.\ref{fig:native-search-3}) using proposed strategy 3.  The results over six campaigns (two from larger spender group, two from medium spender group and two from small spender group) are shown below:

\begin{tabular}{ c|c|c|c|c|c|c}
\hline
\hline
Campaign & 1 & 2 & 3 & 4 & 5 & 6 \\
\hline
ROI &16.01\% &14.13\% & 17.38\% & 21.07\% & 25.74\% & 30.05\% \\
CYR & 60.74\% & 22.88\% & 12.65\% & 13.58\% & 37.49\% & 80.01\% \\
\hline
\hline
\end{tabular}

We observe that for these demonstrated campaigns,  CYR and ROI  increase quite a lot due to the adoption of the recommended bid using the proposed model. In particular, for the small spender campaigns, we have more opportunity to encourage advertisers to increase the bid before spending all the budget to increase CYR and ROI. Essentially, the winning of high quality opportunities using the new bidding strategy contributes the improvement of ROI and CYR.  We notice there exist some campaigns whose ROI increases but CYR decreases (an example is shown below):

{\small
\begin{tabular}{ c|c|c}
\hline
\hline
Campaign & ROI & CYR  \\
\hline
Campaign 1 &12.38\% & -15.48\%\\
\hline
\hline
\end{tabular}
}

The reason is due to the budget limit, the campaign cannot reach the maximum CYR, but instead it needs to decrease the bidding price while achieving reasonable ROI. In our scenario, in practice, it is also practical to choose the inflection point that does not achieve the exact maximum of first-order derivatives (or clicks), but to some level, e.g. 90\% of maximum derivative, where the bid price becomes flatten out (a.k.a ``marketing clearing price''). Statistically, it will be more stable to fluctuations. In addition, it can handle the case when the maximum derivative is obtained at bid = 0.

\subsection{Online A/B test}


We apply A/B test to the ad traffic using two scenarios (1) without any bid optimization using 5\% traffic with default behavior; (2) adopt the infection point as the recommended bid, 5\% traffic. 
In particular, we look at the \emph{bid increase ratio},  \emph{click increase ratio} and \emph{ROI increase ratio}. 
Bid increase ratio (BIR) measures the percentage of bid increase from current bid ($bid^i_{cur}$ for ad $i$) to recommended bid ($bid^i_{IP}$ for ad $i$) using inflection model,  {\it i.e.,}
$
BIR = \sum_{i=1}^n \frac{1}{n}\frac{bid^i_{IP} - bid^i_{cur}}{bid^i_{cur}}. 
$
Click increase ratio (CIR) and ROI increase ratio (RIR) measure the percentage of click and ROI increase from current bid to recommended bid as defined in BIR,  respectively. 
The results are shown below (also in Fig.\ref{fig:bid_dist}): 

{\small
\begin{tabular}{ c|c|c|c}
\hline
\hline
Measurement & BIR& CIR & RIR \\
\hline
Performance liftup &+15.37\% & +30.24\% & +14.50\%\\
\hline
\hline
\end{tabular}
}

\begin{figure}
	\centering
	\includegraphics[height=1.2in, width=0.65\linewidth]{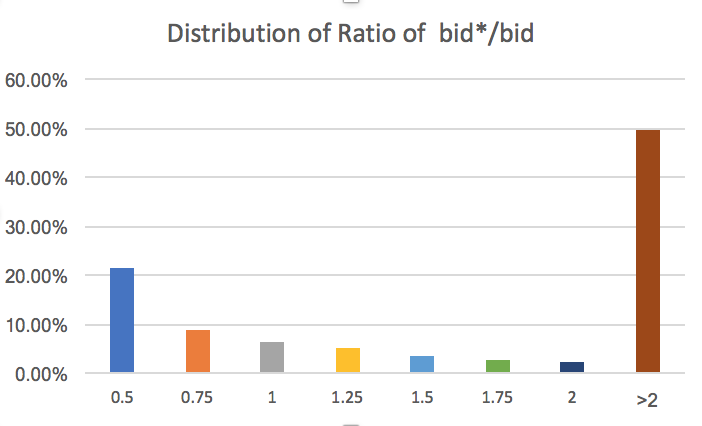}
	\caption{\small Distribution of ratio of optimal $bid^*$ over $bid$. x-axis: ratio of $bid^*$ over $bid$, y-axis: percentage.}
	\label{fig:bid_dist}
\end{figure}
In practice, advertising campaigns may focus on different objective types (e.g., clicks, page-views, conversions, user' experience), which recalls different bidding strategy. This work 
can be easily extended for obtaining inflection point  upon other models.

\vspace{-3mm}
\section{Related Works}
\label{sec:rl}



{\bf Bid recommendation and pricing} In ad auctions,  GSP is the most widely used pricing mechanism\footnote{VCG (Vickrey-Clark-Groves) and  Generalized First Price (GFP).
} due to the approximate truth telling\footnote{However, the market cannot have an equilibrium which may cause the auction oscillation and revenue loss according to the auction theory.} for ad auctions which involves the promotion of keywords by increasing the clicks and impressions for CPC (or CPM) campaigns. 
For real-time bidding (RTB), it has been discovered~\cite{Zhang:2014:ORB:2623330.2623633} that there exist non-linear relations between the optimal bid and ad impressions. Bid optimization and inventory scoring has been widely adopted in targeting ads using supervised learning~\cite{Perlich:2012:BOI:2339530.2339655},  second price auction theory\footnote{\url{https://en.wikipedia.org/wiki/Auction_theory}}, and  linear programming (LP) primal-dual formulation for CPI optimization~\cite{Chen:2011:RBA:2020408.2020604}, {\it etc}.  To allocate the budget for advertisers to achieve the maximum returns, different bidding strategies have been proposed, including
differentiating CTR~\cite{Colini-Baldeschi:2015:MKS:2852252.2818357} in sponsored search, user engagement optimization~\cite{Geyik:2016:JOM:2939672.2939724}, profit maximization at DSP~\cite{2017arXiv170601614G},   dynamically allocating impressions~\cite{DBLP:journals/corr/CaiRZMWYG17}, feedback-control for  risk-aware bidding~\cite{DBLP:journals/corr/ZhangZRRLW17}~\cite{Zhang:2016:FCR:2835776.2835843},  matching of bid and traffic quality~\cite{Zhu:2017:OCP:3097983.3098134} and ~\cite{DBLP:conf/cikm/LinCWC16}, 
 ~\cite{DBLP:conf/aaai/XuSMLQL16}~\cite{Lang:2012:HFE:2187836.2187887},~\cite{2015arXiv151108409F},  {\it etc}. 

 
 Regarding the price design, many pricing methods are proposed to meet the requirement of different application context, including reserve prices design in single item auctions~\cite{DBLP:conf/www/LemePV16},  mixture bidder of GSP and Vickrey-Clarke-Groves (VCG)~\cite{Bachrach:2016:MDM:2872427.2882983}, algorithmic/dynamic pricing~\cite{Chen:2016:EAA:2872427.2883089}, {\it etc}. Unlike real-time bidding, this paper considers the bid decision as an optimization problem by collecting complete campaign information in marketplace  from both demand side and supply side, and each segment of ad volumes are more accurately estimated and optimized to satisfy the advertisers' demands. 
Moreover, the method is designed for both native and search ad campaigns to drive the profit 
from optimally bidding on ad impressions.

{\bf Bid Landscape}
Bid landscape forecasting models the ad market price distributions and winning rate for auctions~\cite{Shah:2017:PES:3097983.3098041} of specific ad inventory given each specific bid price.  Different machine learning models have been proposed to estimate the winning price distributions, for example, 
mixture model~\cite{Wu:2015:PWP:2783258.2783276} for RTB at DSP side,  gradient boosting decision trees~\cite{Cui:2011:BLF:2020408.2020454} on targeting attributes, 
decision tree bid landscape model~\cite{Wang:2016:FBL:3088565.3088573},  Dirac conditioned distribution~\cite{Chapelle:2014:SSR:2699158.2532128}, {\it etc}.  %
In addition, to address the problem of censored data problem~\cite{Zhang:2016:BGD:2939672.2939713} of winning impressions,  Amin {\it et al.}~\cite{DBLP:journals/corr/abs-1210-4847} 
adopted bid-aware gradient descents (BGD) method to %
achieve unbiased learning in tackling strong bias towards the winning impressions. One major drawback of current bid landscape model is that they may not handle the real-world dynamic bid distribution very well due to the divergent and biased bid distributions in real-world settings. In this work, instead of only using the winning price distributions, our bid landscape model~\cite{DBLP:conf/www/GaoKLBY18},~\cite{DBLP:conf/www/KongSY18}, ~\cite{DBLP:conf/www/KongFSY18} is built using both win rate and eCPM\_cost distributions over the bid price from complete market observations where losing price in auctions are properly estimated using GSP, therefore we can perform forecasting more accurately, which largely reduces the model bias and provides accurate price prediction and recommendations.  
The short version of this work appeared in ~\cite{DBLP:conf/www/KongSY18a}.

\section{Conclusion}
\label{sec:conclude}

This paper presents a novel way for bid recommendation using concavity changes, which we believe is an important piece that powers the bid landscape for ad campaign planning and forecasting.  
Rigorous theoretical analysis and extensive experiments on real-world ad campaigns demonstrate the advantages 
over baselines.

{
\bibliographystyle{ACM-Reference-Format}
\bibliography{inflection} 
}
\end{document}